\documentclass[a4paper,11pt]{article}
\pdfoutput=1
\RequirePackage{ifluatex}
\usepackage{fullpage}
\usepackage[utf8]{inputenc}
\usepackage[T1]{fontenc}
\usepackage{graphicx}
\usepackage{amsmath,bm}
\usepackage{float}
\usepackage{amsfonts}
\usepackage{algorithm}      
\usepackage[noend]{algpseudocode}  
\usepackage{algorithmicx}
\usepackage{authblk}
\usepackage[hypertexnames=false,colorlinks=true,urlcolor=Blue,citecolor=Green,linkcolor=BrickRed]{hyperref}
\usepackage[dvipsnames,usenames]{xcolor}
\usepackage{amsthm}
\usepackage[noadjust]{cite}
\usepackage{todonotes}
\usepackage{thmtools,thm-restate}
\usepackage{lineno}
\usepackage{comment}
\usepackage{subfig}

\newcommand{\Oh}{\mathcal{O}}

\newcommand{\eps}{\epsilon}
\newcommand{\lam}{\lambda}

\newcommand{\BR}{\mathsf{BR}}
\newcommand{\BL}{\mathsf{BL}}
\newcommand{\TR}{\mathsf{TR}}
\newcommand{\TL}{\mathsf{TL}}
\newcommand{\last}{\mathsf{last}}
\newcommand{\Blast}{\mathsf{Blast}}
\newcommand{\Bfirst}{\mathsf{Bfirst}}
\newcommand{\Tlast}{\mathsf{Tlast}}
\newcommand{\Tfirst}{\mathsf{Tfirst}}

\newtheorem{theorem}{Theorem}
\newtheorem{lemma}[theorem]{Lemma}
\newtheorem{property}[theorem]{Property}

\title{Optimal Distance Labeling for Permutation Graphs}

\author[ ]{Paweł Gawrychowski\thanks{\texttt{\href{mailto:gawry@cs.uni.wroc.pl}{gawry@cs.uni.wroc.pl}}}}
\author[ ]{Wojciech Janczewski\thanks{\texttt{\href{mailto:wojciech.janczewski@cs.uni.wroc.pl}{wojciech.janczewski@cs.uni.wroc.pl}}}}
\affil[ ]{University of Wrocław}

\date{}

\begin{document}

\maketitle

\begin{abstract}
A permutation graph is the intersection graph of a set of segments between two parallel lines. In other words, they are defined by
a permutation $\pi$ on $n$ elements, such that $u$ and $v$ are adjacent if an only if $u<v$ but $\pi(u)>\pi(v)$. We consider
the problem of computing the distances in such a graph in the setting of informative labeling schemes.

The goal of such a scheme is to assign a short bitstring $\ell(u)$ to every vertex $u$, such that the distance between $u$ and $v$
can be computed using only $\ell(u)$ and $\ell(v)$, and no further knowledge about the whole graph (other than that it is a permutation
graph). This elegantly captures the intuition that we would like our data structure to be distributed, and often leads to interesting
combinatorial challenges while trying to obtain lower and upper bounds that match up to the lower-order terms.

For distance labeling of permutation graphs on $n$ vertices,  Katz, Katz, and Peleg [STACS 2000] showed how to construct
labels consisting of $\Oh(\log^{2} n)$ bits. Later, Bazzaro and Gavoille [Discret. Math. 309(11)] obtained an asymptotically optimal
bounds by showing how to construct labels consisting of $9\log{n}+\Oh(1)$ bits, and proving that $3\log{n}-\Oh(\log{\log{n}})$ bits
are necessary. This however leaves a quite large gap between the known lower and upper bounds.
We close this gap by showing how to construct labels consisting of $3\log{n}+\Oh(\log\log n)$ bits.
\end{abstract}

\thispagestyle{empty}
\clearpage
\setcounter{page}{1}

\section{Introduction}
Geometric intersection graph is a graph where each vertex corresponds to an object in the plane, and two such vertices
are adjacent when their corresponding objects have non-empty intersection. Usually, one puts some restriction on the objects,
for example they should be unit disks. The motivation for such a setup is twofold. First, it allows for modelling many
practical problems. Second, it leads to nice combinatorial questions. This is a large research area, and multiple books/survey
are available~\cite{doi:10.1137/1.9780898719802,pal2013intersection,ellis2022intersection,hlinveny2001representing}
(to name just a few).

In this paper, we are interested in one of the most basic classes of geometric intersection graphs, namely permutation
graphs. A permutation graph is the intersection graph of a set of segments between two parallel lines. An alternative
(and more formal) definition is as follows. A graph $G=(V,E)$, where $V=\{1,2,\ldots,n\}$, is a permutation graph if there
exists a permutation $\pi$ on $n$ elements, such that $u$ and $v$ are adjacent exactly when $u<v$ but $\pi(u)>\pi(v)$.
See Figure~\ref{Fig:PerGraph} for a small example.

\begin{figure}[h]
\begin{center}
  \includegraphics[scale=1.7]{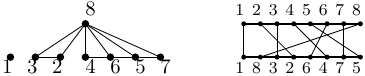}
\end{center}
\caption{Permutation graph described by $\pi= 1 8 3 2 6 4 7 5$.}
\label{Fig:PerGraph}
\end{figure}

Permutation graphs admit a few alternative definitions. For example, $G$ is a permutation graph if and only if
both $G$ and its complement are comparability graphs~\cite{DBLP:journals/jacm/EvenPL72}. Alternatively, they can
be defined as comparability graphs of two-dimensional posets~\cite{DBLP:journals/networks/BakerFR72}.
From the algorithmic point of view, the motivation for studying such graphs is that they can be recognised in linear
time~\cite{McConnellS99}, and multiple problems that are computationally difficult on general graphs admit
efficient algorithms on permutation graphs~\cite{Moehring1985,CHAO2000159,Colbourn}.
In this paper, we consider constructing a distributed data structure capable of efficiently reporting the distance between
two given vertices of a permutation graph.

\paragraph{Informative labeling schemes. } We work in the mathematically elegant model of informative
labeling schemes, formally introduced by Peleg~\cite{Peleg05}. Such a scheme is meant to represent graphs in an extremely distributed way.
Instead of storing a single global data structure,
a scheme assigns to each vertex $v$ of a given graph a binary string $\ell(v)$, called a label.
Later, given the labels of two vertices (and no additional information about the graph), we should be able to compute
some fixed function on those two vertices.

In the context of informative labeling schemes, the first function that one usually considers is adjacency,
where we simply want to decide whether the two vertices in question are neighbours in the graph.
As observed by Kannan, Naor, and Rudich~\cite{Kannan}, this is equivalent to finding a so-called vertex-induced
universal graph, and predates the more general notion of informative labeling schemes.
Non-trivial adjacency labeling schemes have been constructed for many classes of graphs,
for example undirected, directed, and bipartite graphs~\cite{alstrup2015adjacency},
graphs of bounded degree~\cite{EsperetLO08}, trees~\cite{alstrup2015optimal}, planar graphs~\cite{Planar1,PlanarMain},
comparability graphs~\cite{BonamyEGS21}, or general families of hereditary graphs~\cite{HatamiH22,Monotone}.
In every case, the length of each \emph{individual} label is much smaller than the size of a centralised structure,
often by a factor close to $\Theta(n)$, i.e., we are able to evenly distribute the whole adjacency information.
Other functions considered in the context of labeling schemes
are  ancestry in trees~\cite{FraigniaudAncestry,GawrychowskiKLP18},
routing~\cite{thorup2001compact,FraigniaudG01,GawrychowskiJL21} or connectivity~\cite{Korman10}.
However, from the point of view of possible applications, the next most natural question is that of distance labelings,
where given labels of two vertices we need to output the exact distance between them in a graph.
This properly generalises adjacency and usually needs much longer labels.

\paragraph{Distance labelings.}
The size of a labeling scheme is defined by the maximum length of any label assigned by the encoder.
If not stated otherwise, all graphs are unweighted and undirected, and consist of $n$ vertices.
For general undirected graphs, Alstrup, Gavoille, Halvorsen, and Petersen~\cite{alstrup2016simpler} constructed distance
labeling of size $(\log{3})n/2 + o(n)$, while the known lower bound is $\lceil n/2 \rceil$ bits.
Alstrup, Dahlgaard, Knudsen, and Porat~\cite{AlstrupDKP16} describe a slightly sublinear $o(n)$-bits labeling for sparse graphs. 
In case of planar graphs, scheme of size $\Oh(\sqrt{n})$ bits is presented by Gawrychowski and Uznanski~\cite{GawrychowskiU16},
and the known lower bound is $\Omega(n^{1/3})$ bits.
Shur and Rubinchik~\cite{ShurR24} designed a scheme using $n^{1.5}/\sqrt{6}+\Oh(n)$ distinct labels for families of cycles,
against a lower bound of $\Omega(n^{4/3})$~\cite{DBLP:journals/algorithmica/KormanPR10}.
For trees, we do not need a polynomial number of bits, as they can be labeled for distances using only
$1/4 \log^2{n}+o(\log^2{n})$ bits as shown by Freedman, Gawrychowski, Nicholson, and Weimann~\cite{FGNW16},
which is optimal up to the second-order terms~\cite{alstrup2015distance}. Of course, the interesting
question is to find natural classes of graphs that admit small distance labeling schemes.

\paragraph{Distance labeling for permutation graphs.}
Katz, Katz and Peleg~\cite{KatzKP05} presented distance labeling scheme of size $\Oh(\log^2{n})$
for interval and permutation graphs.
This was improved by Gavoille and Paul to $5\log{n}$ labeling for interval graphs~\cite{GavInterval},
with a lower bound of $3\log{n}-\Oh(\log\log n)$.
Very recently, He and Wu~\cite{HeW24} presented tight $3\log{n}+\Oh(\log{\log{n}})$ distance labeling for interval graphs.
For connected permutation graphs, Bazzaro and Gavoille in~\cite{BazzaroG05} showed a distance labeling scheme of size $9\log{n}+\Oh(1)$ bits,
and a lower bound of $3\log{n}-\Oh(\log\log n)$.
As noted in their work, this is especially interesting as there are very few hereditary graph classes that admit
distance labeling schemes of size $o(\log^2{n})$.
As our main result, we close the gap between the lower and upper bounds on the size of distance labeling for
permutation graph, by showing the following theorem.

\begin{theorem}
There is a distance labeling scheme for permutation graphs with $n$ vertices using labels of size $3\log{n}+\Oh(\log{\log{n}})$ bits.
The distance decoder has constant time complexity, and labels can be constructed in polynomial time.
\label{Th:Main}
\end{theorem}

\paragraph{On constants.}
We stress that in the area of informative labeling scheme, it is often relatively easy to obtain asymptotically optimal
bounds on the size of a scheme, and the real challenge is to determine the exact constant for the higher-order term.
This has been successfully done for multiple classes, e.g. distance labeling for trees, where $\Oh(\log^2{n})$~\cite{PelegTD} was
first improved to  $(1/2)\log^2{n}$~\cite{alstrup2015distance} and then $(1/4)(\log^2{n})+o(\log^2{n})$~\cite{FGNW16},
optimal up to second-order term.
Adjacency labeling for trees is a particularly good example, with the first scheme having a size of $6\log{n}$
based on~\cite{Muller}, then $4\log{n}$~\cite{Kannan}, $(2+o(1))\log{n}$~\cite{GavoilleL07},
$(4/3+o(1))\log{n}$~\cite{Planar1}, finally $\log{n}+\Oh(\sqrt{\log{n}\log{\log{n}}})$~\cite{PlanarMain}
and $\log{n}+\Oh(\sqrt{\log{n}})$~\cite{GawrychowskiJ22} were presented, the last two being optimal up to the second order terms.
For adjacency in bounded-degree graphs with odd $\Delta$, initial $(\Delta/2+1/2)\log{n}+\Oh(1)$~\cite{Butler2009} was improved
to $(\Delta/2+1/2-1/\Delta)\log{n}+\Oh(\log{\log{n}})$~\cite{EsperetLO08} and then to optimal $(\Delta/2)\log{n}+\Oh(1)$~\cite{AlonN17}.
In the case of adjacency labelings for general undirected graphs, starting with the classical result presenting labels of size
$n/2+\Oh(\log{n})$~\cite{moon_1965}, $n/2+\Oh(1)$~\cite{alstrup2015adjacency} and $n/2+1$~\cite{Alon17}
labelings were constructed.
Similar sharply optimal labelings are shown for directed graphs, tournaments, bipartite graphs, and oriented graphs.
Finally, the first described ancestry labeling schemes for trees was of size $2\log{n}$~\cite{Kannan}, and then
$(3/2)\log{n}$~\cite{AbiteboulKM01}, $\log{n}+\Oh(\log{n}/\log{\log{n}})$~\cite{thorup2001compact}, $\log{n}+\Oh(\sqrt{\log{n}})$~\cite{AbiteboulAKMR06}, $\log{n}+4\log{\log{n}}+\Oh(1)$~\cite{FraigniaudAncestry}, $\log{n}+2\log{\log{n}}+\Oh(1)$~\cite{DahlgaardKR15} schemes were provided, achieving optimality up to second-order terms.

\paragraph{Related works.}
The challenge of designing labeling schemes with short labels is related to that of designing succinct data structures,
where we want to store the whole information about the input (say, a graph) using very few bits, ideally at most the information
theoretical minimum. This is a rather large research area, and we only briefly describe the recent results on succinct data structures
for the interval and permutation graphs.
Tsakalidis, Wild, and Zamaraev~\cite{TsakalidisWZ23} described a structure using only $n\log{n}+o(n\log{n})$ bits (which is optimal)
capable of answering many types of queries for permutation graphs.
They also introduce the concept of semi-distributed representation, showing that for distances in permutation graphs
it is possible to store global array of size $\Oh(n)$ bits and labels on only $2\log{n}$ bits,
offering a mixed approach which can overcome $3\log{n}$ lower bound for distance labeling.
For interval graphs, a structure using $n\log{n}+\Oh(n)$ bits (which is again optimal) is known (\cite{AcanCJS21} and~\cite{HMNWW20}).

\section{Overview and organisation}

In Section~\ref{Sec:Prelim}, we present basic definitions for labeling schemes and our approach to permutation graphs.
Then, in  Section~\ref{Sec:5log}, we build on the methods of Gavoille and Paul~\cite{GavInterval}, as well as
Bazzaro and Gavoille~\cite{BazzaroG05} for creating distance labelings of interval and permutation graphs.
We can represent a permutation graph as a set of points in the plane, where two points (two vertices) are adjacent when
one is above and to the left of the other.

The first thing we need to notice when considering distances is the presence of two \emph{boundaries} in such representation.
We say that the top boundary is formed by points with empty (containing no other points) top-left (north-west) quadrant,
and bottom boundary by points with empty bottom-right (SE) quadrant.
Points on the boundaries are especially important -- it can be seen that for any pair of points, there is a shortest path between them with
all internal points of the path being on boundaries.
As a set of boundary points forms a bipartite graph, such shortest path strictly alternates between boundaries.

We can also observe that for a point $v$ not a boundary, there are four boundary points of special interest for it, see Figure~\ref{Fig:Over1}.
These are pairs of extreme points on both boundaries from the points adjacent to $v$.
Any shortest path from $v$ to $u$ with $d(v,u)>2$ can have as a second point one of these special points for $v$, and as a penultimate point one of the special points for $u$.
We need to handle distances of 1 and 2 separately, but otherwise, this means it is enough to be able to compute distances between boundary points.

\begin{figure}[h]
\begin{center}
  \includegraphics[scale=1]{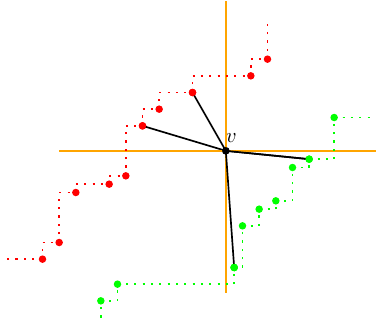}
\end{center}
\caption{Green points form bottom boundary, red points top boundary.
$v$ is not on the boundary, orange lines show its quadrants.
For $v$, there are four points of special interest, extreme neighbours on both boundaries.}
\label{Fig:Over1}
\end{figure}

If we can build distance labeling for boundary points, and store labels of special points efficiently, we can obtain good distance labelings for permutation graphs.
This is possible as boundaries are highly structured, in particular ordered.
In~\cite{BazzaroG05} authors view two boundaries as two proper interval graphs and deal with them using methods from~\cite{GavInterval}.
An interval graph is proper when no interval is completely contained by another one.
Gavoille and Paul first partition vertices of a proper interval into distance \emph{layers}, by distances to the vertex
representing the leftmost interval.
Let us denote layer number of vertex $u$ by $L(u)$.
It can be seen that for any two vertices $u,v$ in interval graph we have either $d(u,v)=|L(u)-L(v)|$ or $d(u,v)=|L(u)-L(v)|+1$.
Then the following lemma is used~\cite{GavInterval}:

\begin{lemma}
There exists a total ordering $\lam$ of vertices of proper interval graph such that given $\lam(v), \lam(u)$
and layer numbers $L(u) < L(v)$ for two vertices $u,v$, we have $d(u,v)=L(v)-L(u)$
if and only if $\lam(u) > \lam(v)$.
\end{lemma}

In other words, we can assign to each vertex $v$ just two numbers $L(v), \lam(v)$,
and then still be able to determine all exact distances.
Going back to permutation graphs, when we view two boundaries as proper interval graphs, it is possible to obtain
straightforward distance labeling for permutation graphs using $20\log{n}$ bits,
where the big constant is due to storing many distance labels for interval graphs completely independently.
Then authors are able to reduce the size of labels to $9\log{n}$ bits, after eliminating many redundancies in the stored sub-labels.

\begin{figure}[h]
\begin{center}
  \includegraphics[scale=1]{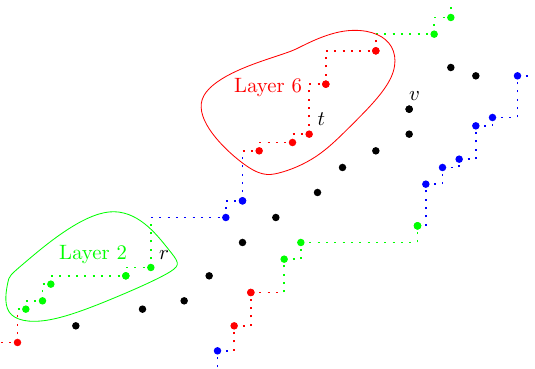}
\end{center}
\caption{Boundary points partitioned into layers. $r,t$ are on the top boundary, but in layers 2 and 6.
For any two points in layers $a$ and $b$, the distance between them is always either $|a-b|$ or $|a-b|+2$;
here, $d(r,t)=4$.
$v$ is not on the boundary, and any such point can be adjacent to points from at most three different layers.}
\label{Fig:Over2}
\end{figure}

In this paper, we show that working with both boundaries at once can yield better results.
To do this, we modify the methods of Bazzaro and Gavoille and then carefully remove even more redundancies.
First, we partition points on both boundaries into layers, defined by distances from some initial point,
see Figure~\ref{Fig:Over2}.
As we use distances from a single point to define layers, the distance between any two boundary points is a difference
of their layer numbers, or this value increased by two.
It can be shown that again some $\lam$ ordering can be used, and storing it takes around $\log{n}$ bits for each boundary point.

As a single point is adjacent to at most three layers, layer numbers of four special points are easy to store,
and we could achieve labeling of length $(2+1+4)\log{n}+\Oh(1)=7\log{n}+\Oh(1)$ by storing for each point respectively
its 2D coordinates, layer numbers of neigbours and four times $\lam$ values for extreme neighbours,
in order to compute distances between boundary points.
This can be reduced to $5\log{n}$ by dealing with distances 1 and 2 more carefully, allowing us to not store point coordinates explicitly.
All of the above is described in Section~\ref{Sec:5log}.

After additional analysis and reductions laid out in Section~\ref{Sec:3log}, we can decrease the size to $3\log{n}$.
This is since, roughly speaking, one can collapse information stored for two pairs of extreme
boundary neighbours into just two numbers, due to useful graph and layers properties.
More precisely, we can observe that we store excessive information about the set of four extreme neighbours.
For vertex $v$, two extreme right points on both boundaries are used to reach points to the right of $v$,
and extreme left are used to reach points to the left.
But we do not need the exact distance between points to the left of $v$ and right extreme points, thus we have some possibility
to adjust the stored $\lam$ values.
Particularly, the main case is when $\lam$ value of the right extreme point on the bottom boundary is smaller than 
$\lam$ value of the left extreme point on the top boundary;
it turns out that these two values can be equalised and stored as some single value in between the original values.
The second pair of extreme points can be dealt with in a similar manner, and then we need to ensure that all of this
did not interfere with the correctness of deciding about distances 1 and 2, which are different cases than all distances
larger than 2.

\section{Preliminaries}
\label{Sec:Prelim}

\paragraph{Permutation graphs.}
Permutation graph $G_{\pi}$ is a graph with vertices representing elements of permutation $\pi$,
where there is an edge between two vertices if and only if the elements they represent form an inversion in the permutation.
See Figure~\ref{Fig:PerGraph}.
In~\cite{McConnellS99} McConnell and Spinrad show that it is possible to test in linear time whether
a given graph is a permutation graph, and also construct the corresponding permutation.

We will use a geometric (or 'grid') representation of permutation graph $G_{\pi}$ on $n$ vertices as a set of points with coordinates in $[1,n]$,
with point $(i,\pi^{-1}(i))$ for each $i \in [1,n]$.
Considering a point $p$, we always denote its coordinates by $p=(p_x,p_y)$.
Top-left quadrant of point $p$, $\TL_p$, is a subset of points $\{v: v_x < p_x \land v_y > p_y\}$ from the graph.
Similarly, we have $\TR_p$ (top-right), $\BL_p$ (bottom-left) and $\BR_p$ (bottom-right) quadrants.
Two points are adjacent in the graph iff one is in $\TL$ or $\BR$ quadrant of the other.
See Figure~\ref{Fig:Grid}.
We have transitivity in a sense that if $w \in \BR_v$ and $u \in \BR_w$, then $u \in \BR_v$; similarly for other quadrants.
By distance $d(u,v)$ between two points we will mean distance in the graph.

\begin{figure}[h]
\begin{center}
  \includegraphics[scale=1.7]{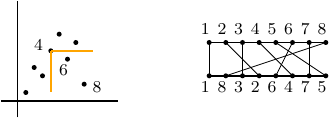}
\end{center}
\caption{Geometric representation of the graph from Figure~\ref{Fig:PerGraph}, so permutation [1,8,3,2,6,4,7,5].
Point $(4,6)$ is adjacent to $(6,5)$ and $(8,2)$, as these two points are in its bottom-right quadrant,
while top-left quadrant of $(4,6)$ is empty.}
\label{Fig:Grid}
\end{figure}

We will assume the given permutation graph is connected.
There are standard ways to enhance labelings to disconnected graphs by adding at most $\Oh(\log{\log{n}})$ bits to the labels,
and we will describe how it can be done after the main theorem.
We note that for connected graphs of size at least two, no point could be on both boundaries,
as it would be isolated otherwise.

\paragraph{Labeling schemes.}
Let $\mathcal{G}$ be a family of graphs.
A distance labeling scheme for $\mathcal{G}$ consists of an encoder and a decoder.
The encoder takes a graph $G\in \mathcal{G}$ and assigns
a label (binary string) $\ell(u)$ to every vertex $u\in G$.
The decoder receives labels $\ell(u)$ and $\ell(w)$,
such that $u,w\in G$ for some $G\in \mathcal{G}$ and $u \neq w$,
and should report the exact distance $d(u,w)$ between $u$ and $w$ in $G$.
The decoder is not aware of $G$ and only knows that $u$
and $w$ come from the same graph belonging to $\mathcal{G}$.
We are interested in minimizing the maximum length of a label, that is, $\max_{G\in \mathcal{G}}\max_{u\in G} |\ell(u)|$.

\paragraph{Organization of the labels.}
The final labels will consist of a constant number of parts.
We can store at the beginning of each label a constant number of pointers to the beginning of each of those parts.
As the total length of a label will be $\Oh(\log{n})$, pointers add only $\Oh(\log{\log{n}})$ bits to the labels.

\section{Scheme of Size 5$\log{n}$}
\label{Sec:5log}
In this section, we describe how to use boundaries to design distance labeling of size $7\log{n}+\Oh(1)$,
and then how to refine it to reach $5\log{n}+\Oh(1)$.

\subsection{Properties of Boundaries}

\begin{figure}[h]
\begin{center}
  \includegraphics[scale=0.7]{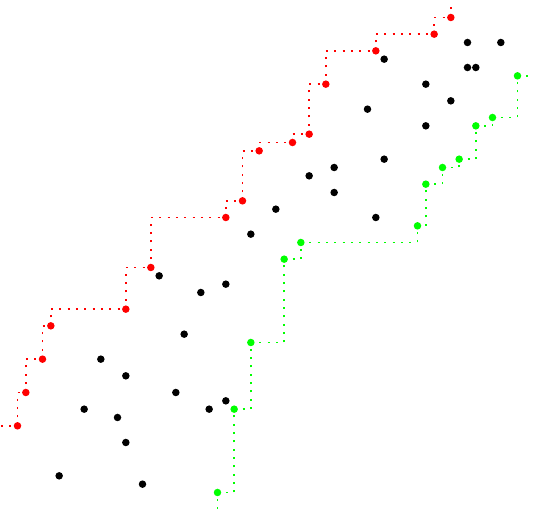}
\end{center}
\caption{Green points are on the bottom boundary, red points are on the top boundary.}
\label{Fig:Boundaries}
\end{figure}

For a set of points $S$, we have its top boundary defined as a subset of points from $S$ which top-left quadrants are empty,
and bottom boundary as a subset of points which bottom-right quadrants are empty.
See Figure~\ref{Fig:Boundaries}.
Observe that points on boundaries are ordered, that is, for $u$ and $v$ on the same boundary, either $u_x>v_x$ and $u_y>v_y$,
or $u_x<v_x$ and $u_y<v_y$.
We use $<$ to note this relation on boundary points.

\begin{figure}[h]
\begin{center}
  \includegraphics[scale=0.7]{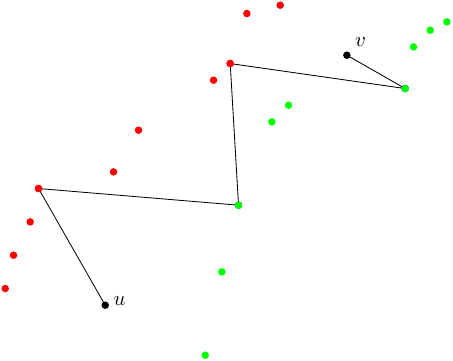}
\end{center}
\caption{Path between two points alternating between top and bottom boundaries.}
\label{Fig:AltPaths}
\end{figure}

Boundaries are particularly useful when considering distances between points:
\begin{property}
For any two points $u,v$ at distance $d$, there is a path $P=(u=q_0,q_1,q_2,\ldots,q_d=v)$ of length $d$
such that all points except possibly $u,v$ are on alternating boundaries.
\label{Prop:Bounds}
\end{property}
\begin{proof}
Take any shortest path $P'$ and any adjacent $q_i$ and $q_{i+1}$ on $P'$, assume without loss of generality that $q_{i+1} \in \TL_{q_i}$.
Suppose that $q_{i+1}$ is not on the top boundary.
We either have $q_{i+2} \in \TL_{q_i+1}$ or $q_{i+2} \in \BR_{q_i+1}$.
Note that by transitivity if $q_{i+2} \in \TL_{q_{i+1}}$, then $q_{i+2} \in \TL_{q_i}$ and we could have a shorter path by removing $q_{i+1}$.
Thus, assume $q_{i+2} \in \BR_{q_{i+1}}$.
If $q_{i+1}$ is not on the top boundary, then by definition there exists a boundary point $q' \in \TL_{q_{i+1}}$,
and it must be that $q_{i+2} \in \BR_{q'}$.
This means that we could replace $q_{i+1}$ by $q'$, increasing the number of points from the path lying on the boundary, and then repeat the argument.
Therefore, we have that all points except the first and last ones can always lie on boundaries.
See Figure~\ref{Fig:AltPaths} for an illustration.
These must be alternating boundaries, as by definition no two points on the same boundary are adjacent.
\end{proof}

\begin{figure}[h]
\begin{center}
  \includegraphics[scale=0.8]{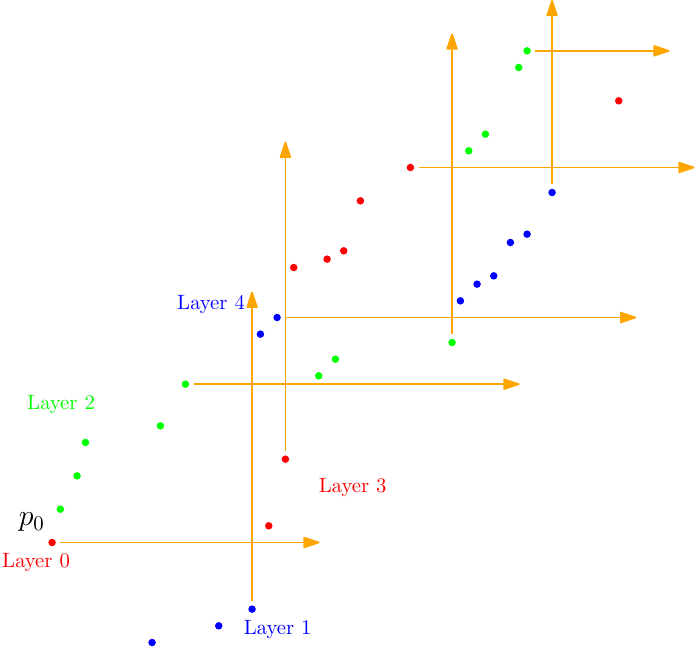}
\end{center}
\caption{Layers on boundary points defined as distances from the leftmost point $p_0$.
Orange lines represent which points are adjacent to (being in $\BR$ or $\TL$) the last point in the layer.}
\label{Fig:Layers}
\end{figure}

We partition all points on the boundaries into layers, in the following way.
The layer number $0$ consists of a single left-most point $p_0$ in the whole set $S$.
Note that $p_0$ is on the top boundary.
Then, a boundary point is in a layer number $i$ if its distance to $p_0$ is $i$.
By $L(u)$ we denote the layer number of $u$.
See Figure~\ref{Fig:Layers}, we will soon see that indeed layers are always nicely structured, as pictured.
Observe that in even layers, we have only points from the top boundary, and in odd layers only from the bottom boundary,
as points on a single boundary are non-adjacent.
Thus, points in a single layer are ordered, by both coordinates.

To determine the distance between boundary points, we use a method similar to the one from the paper of
Gavoille and Paul~\cite{GavInterval}, precisely Theorem 3.8.
This is also connected to what Bazzaro and Gavoille~\cite{BazzaroG05} do in their work, but not identical,
as they use bottom and top boundaries separately, as two mostly independent interval graphs.

\begin{lemma}
There exists a total ordering $\lam$ of boundary points such that given $\lam(v), \lam(u)$
and layer numbers $L(u) < L(v)$ for two boundary points $u,v$, we have $d(u,v)=L(v)-L(u)$
if and only if $\lam(u) > \lam(v)$.
\label{Lem:Lambda}
\end{lemma}
\begin{proof}
As noted, points on both boundaries are ordered, and layers switch between boundaries,
starting with layer number 0 containing just a single left-most point from the top layer.
Say ordered points on the top layer are $t_0,t_1,t_2,\ldots$.
We prove that there exist strictly increasing numbers $i_0=0,i_2,i_4,\ldots$ such that layer number 0 consists of $t_0$,
and then any layer $2k$ consists of consecutive points $t_{i_{2k-2}+1},\ldots,t_{i_{2k}}$.
Similarly, for points $b_0,b_1,\ldots$ on bottom layer, there exists numbers $i_1,i_3,\ldots$ defining analogous ranges.
Denote by $\last(q)$ the last point (with largest coordinates) in layer $q$.
We prove by induction, in a given order, some intuitive properties (considering without loss of generality odd layer):
\begin{enumerate}
\item All points from layer $2k+1$ are to the right of all points from layer $2k$ (or all points from layer $2k$ are above layer $2k-1$).
\item All points from layer $2k+1$ are adjacent to $\last(2k)$.
\item Layer $2k+1$ is formed by consecutive points $b_{i_{2k-1}+1},\ldots,b_{i_{2k+1}}$.
\item Ordered points from layer $2k$ are adjacent to increasing prefixes of points $b_{i_{2k-1}+1},\ldots,b_{i_{2k+1}}$.
This means that, firstly, any point $t_j$ with $L(t_j)=2k$ is adjacent exactly to points $b_{i_{2k-1}+1}, \ldots, b_q$
from layer $2k+1$, for some $q \leq i_{2k+1}$.
Secondly, for $t_{j+1}$ with $L(t_{j+1})=2k$, $t_{j+1}$ is adjacent to points $b_{i_{2k-1}+1}, \ldots, b_r$ with $q \leq r$.
\end{enumerate}
The base of layer 0 is apparent, except for the third property, which can be done as in the induction step.
Now consider layer $2k+1$.
As all points from layer $2k$ are adjacent to $\last(2k-1)$, meaning they are in $TL_{\last(2k-1)}$,
all these points are to the left of points from layer $2k+1$.
Moreover, the last point in layer $2k$ has the largest $y$ coordinate, thus if any point from layer $2k+1$
is adjacent to some point from layer $2k$, then it is also adjacent to $\last(2k)$.
This gives us the first two properties.

\begin{figure}[h]
\begin{center}
  \includegraphics[scale=1]{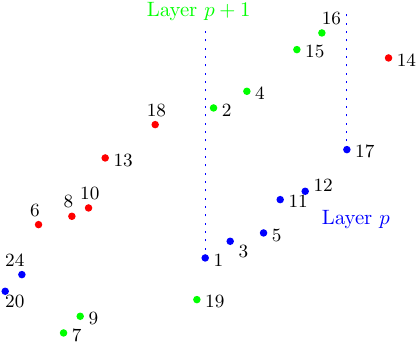}
\end{center}
\caption{Considering layers $p$ and $p+1$, a point with $\lam$ value of 1 is not adjacent to any point from layer $p+1$,
an empty prefix.
The point with value 3 is adjacent to just point with value 2, so range $[2,2]$.
5 and 11 are adjacent to prefix $[2,4]$.
12 is adjacent to $[2,15]$, and finally, 17 is adjacent to all points from layer $p+1$.}
\label{Fig:Prefixes}
\end{figure}

Now, by definition, if points $b_q, b_r$ with $r>q$ are adjacent to some point $v$ on the top boundary,
then all points $b_q, b_{q+1}, \ldots, b_r$ are adjacent to $v$.
Thus, if $b_{i_{2k-1}+1}$ is adjacent to $\last(2k)$, we get the third property.
But it must be adjacent, as otherwise it would be above $\last(2k)$, and then layer $2k+1$ would be empty.
We can apply the above principles to any point from layer $2k$ - it either neighbours the first point in layer $2k+1$ or no point from this layer.
This means any point from layer $2k$ neighbours prefix of points from layer $2k+1$ (possibly empty, possibly full layer).
Moreover $t_{j+1}$ is adjacent to the same points from layer $2k+1$ that $t_j$ is, and possibly more, as $t_{j+1}$ is above $t_j$.
See Figure~\ref{Fig:Prefixes}.

\begin{figure}[h]
\begin{center}
  \includegraphics[scale=1]{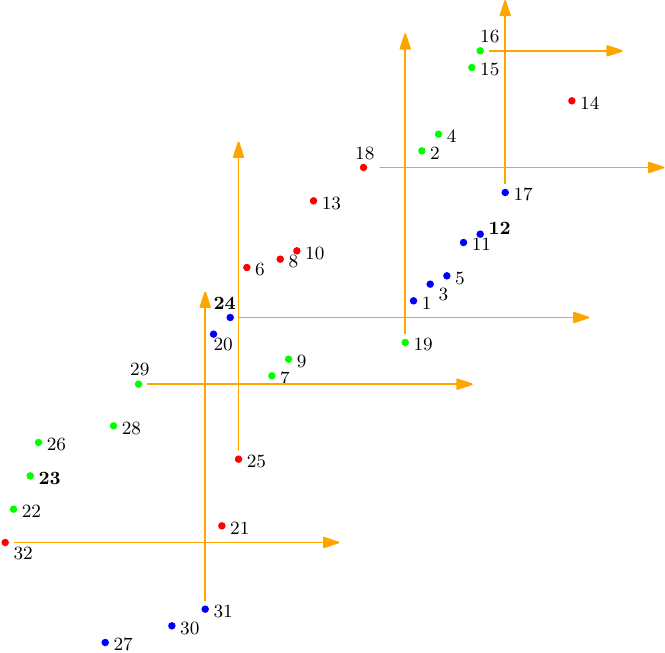}
\end{center}
\caption{Boundary points with their $\lam$ values. The point with value 23 can reach point 12 using five edges,
by path 23-21-20-19-18-12.
But 23 cannot reach 24 using two edges, it needs four.}
\label{Fig:Lambda}
\end{figure}

By the definition of layers, for points $u,v$, $d(u,v) \geq |L(u)-L(v)|$.
If $d(u,v)=|L(u)-L(v)|$, we say that there is a \emph{quick path} between them.
Going back to the statement of the lemma, it says there is such ordering $\lam(v)$ applied to boundary points,
that there is a quick path between two points if relations between their $\lam$ values and layer numbers are opposite.
We can create ordering $\lam$ in the following greedy way: starting from $\lam$ value of 1,
always assign current value to the lowest point in the lowest layer such that it is adjacent to no points in the next layer
without $\lam$ values already assigned, then increment current value.
In other words, repeatedly choose the lowest (by layer) possible point
having all neighbours from the next layer already assigned $\lam$ values.
See Figure~\ref{Fig:Lambda} for an example.

For correctness, observe that when we choose $v$ from layer $k$,
for all layers larger than $k$ the last point with assigned value is adjacent to all points
with assigned values in the next layer.
This is by greedy procedure, as assume there is a layer $j>k$ such that $u$, the last point with assigned value in layer $j$,
is not adjacent to $w$, the last point with assigned value in layer $j+1$.
It is only possible if $\lam(w)>\lam(u)$.
But by definition of greedy procedure, there is no reason to choose $w$ at any point after choosing $u$
and before choosing $v$, as no other point from layer $j$ was chosen between these events and procedure always choose
the lowest layer.
Now, using the above, we can observe that at the moment we assign value for point $v$ in layer $k$:
\begin{itemize}
\item There are quick paths from $v$ to all points with assigned values and in layers larger than $k$.
This is true by the choice of $v$ and transitivity. It is given that $v$ is adjacent to all points in layer $k+1$
with assigned values, then the last of these points is adjacent to all points in layer $k+2$ with assigned values, and so on.
\item There are no quick paths from $v$ to any point without assigned value in a layer larger than $k$.
This is clear from the greedy procedure, previous point and the fourth inductive property -
we do have quick paths to points in larger layers up to the last point with an assigned value, and these points cannot be
adjacent to any more points, since they were chosen only when all their neighbours from the next layer got assigned values. 
\qedhere
\end{itemize}
\end{proof}

By Lemma~\ref{Lem:Lambda}, we are able to detect when quick paths exist, and to complete knowledge about distances between boundary points we observe the following:
\begin{property}
For any boundary points $u,v$ with $L(v) \geq L(u)$, $d(u,v)$ is equal to either $L(v)-L(u)$ or $L(v)-L(u)+2$.
\label{Prop:Diff2}
\end{property}
\begin{proof}
First let us argue that $d(u,v) \leq L(v)-L(u)+2$.
We observed $\last(i)$ is adjacent to all points in layer $i+1$.
Thus, $(u,\last(L(u)-1),\last(L(u)),\ldots,\last(L(v)-1),v)$ is always a correct path with length $L(v)-L(u)+2$.
By definition of layers, $d(u,v)$ cannot be less than $L(v)-L(u)$.
Finally, by the Property~\ref{Prop:Bounds} there is a shortest path that alternates between boundaries,
so it cannot be of length $L(v)-L(u)+1$, as we cannot change parity.
\end{proof}

To simplify our proofs, we will add some points to the original set.
For each point $v$ on the bottom boundary and from the original input, we add point $(v_x+\eps, v_y-\eps)$.
It is easy to see that such a point lies on the bottom boundary, adding it does not change distances between any
existing points, and it removes $v$ from the bottom boundary.
Then we change numeration to integer numbers again, increasing the range of numbers by some constant factor.
Similarly, for any original point $v$ on top boundary, we add $(v_x-\eps, v_y+\eps)$.
What we achieve is that after this change no point from the original input lies on the boundary,
which reduces the number of cases one needs to consider when assigning labels (only to the original points).

\begin{figure}[h]
\begin{center}
  \includegraphics[scale=1]{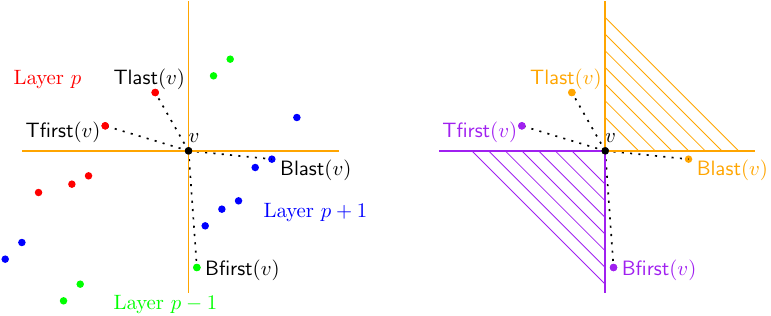}
\end{center}
\caption{On the left: boundary points adjacent to $v$. Orange lines represent $\TL_v$ and $\BR_v$.
$v$ is adjacent to the two last points in red layer $p$, one last point in green layer $p-1$, and the first five points in layer $p+1$.
On the right: Two orange extreme neighbours are needed for paths to points in marked $\TR_v$, and two purple ones for $\BL_v$.}
\label{Fig:Vision}
\end{figure}

\begin{figure}[h]
\begin{center}
  \includegraphics[scale=1]{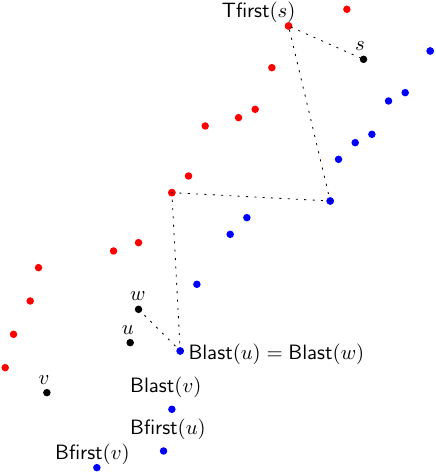}
\end{center}
\caption{We have $d(u,v)=2$, even though all $\Bfirst(v),\Blast(v),\Bfirst(u),\Blast(u)$ are different, but their ranges do intersect.
Moreover, there is a shortest path from $w$ to $s$ with the second point being $\Blast(w)$ and the penultimate point being $\Tfirst(s)$.}
\label{Fig:Dist}
\end{figure}

Now let us focus on any point $v$ not on the boundary.
Assume $v$ is adjacent to some points from layers $i$ and $j$, $j>i$.
It cannot be that $j>i+2$, by definition of layers as distances from $p_0$.
Thus, $v$ is adjacent to points from at most three (consecutive) layers.
We note that $v$ is adjacent to a consecutive segment of ordered points from any layer $i$.
Let us denote by $\Bfirst(v)$ and $\Blast(v)$ the first and last points on the bottom boundary adjacent to $v$
and by $\Tfirst(v)$ and $\Tlast(v)$ the first and last points on the top boundary adjacent to $v$.
Consult Figure~\ref{Fig:Vision}.

We can make easy observation on points at distance two:
\begin{property}
For any two points $u,v$, $d(u,v) \leq 2$ is equivalent to \newline $[\Bfirst(u),\Blast(u)] \cap [\Bfirst(v),\Blast(v)] \neq \emptyset$
or $[\Tfirst(u),\Tlast(u)] \cap [\Tfirst(v),\Tlast(v)] \neq \emptyset$.
\label{Prop:Dist2}
\end{property}
This is since by Property~\ref{Prop:Bounds}, we must have a path between $u,v$ at distance two going through a single point on the boundary.
In the case of $d(u,v)=1$, assume without loss of generality that $u \in TL_v$, then
$[\Tfirst(u),\Tlast(u)] \subseteq [\Tfirst(v),\Tlast(v)]$, and ranges are never empty.

Considering points at a distance of at least three, we have the following:
\begin{lemma}
For any two non-boundary points $u,v$ with $d(u,v)>2$ and $u_x<v_x$, there is a shortest path from $u$ to $v$ with the second point being
either $\Blast(u)$ or $\Tlast(u)$ and the penultimate point being either $\Bfirst(v)$ or $\Tfirst(v)$.
\label{Lem:Dist3}
\end{lemma}
\begin{proof}
We will prove statement for the second point, as the penultimate point is symmetric.
By Property~\ref{Prop:Bounds} there always exists a shortest path $P$ with all but extreme points lying on alternating boundaries,
with $P=(u=q_0,q_1,q_2,\ldots,w,q_{d(u,v)}=v)$, so we denote penultimate point by $w$.

Consider layer number of $w$.
As $d(u,v)>2$ and $u_x<v_x$, it must be that $v,w \in TR_u$ and so $L(w) \geq \min(L(\Tlast(u)),L(\Blast(u)))$.
If $L(w) \geq \max(L(\Tlast(u)),L(\Blast(u)))$, then by Property~\ref{Prop:Diff2} and Lemma~\ref{Lem:Lambda}
we can replace $q_1$ with $\Blast(u)$ or $\Tlast(u)$ (which have the largest $\lam$ values in their layers from
neighbours of $w$), while keeping the length of $P$ and $w$ as penultimate point.
We are left with $L(w) = \min(L(\Tlast(u)),L(\Blast(u)))$.
Assume $L(\Blast(u))>L(\Tlast(u))$, so $L(w)=L(\Tlast(u)$ and also $w>\Tlast(u)$.
Then $w$ is adjacent to $\last(L(w)-1) \in \BR_u$ and thus also to $\Blast(u)$, so we can have $q_1=\Blast(u)$.
In other case, we could similarly set $q_1=\Tlast(u)$.
Thus, we can always change $q_1$ to be $\Blast(u)$ or $\Tlast(u)$, without changing $w$.
\end{proof}

We established ways to determine the distance between any points using distances between specific boundary points.
Additionally, observe that all conditions from Lemma~\ref{Lem:Lambda} and Property~\ref{Prop:Dist2}
can be checked using just $\lam$ values and layer numbers.
That is, for $u,v$ on the same boundary we have $u \leq v$ iff $(L(u),\lam(u)) \leq_{lex} (L(v),\lam(v))$. 

At this stage, we could create labels of length $7\log{n}+\Oh(1)$, by storing for each point $v$ coordinates 
$v_x,v_y$, and for all points of interest $\Bfirst(v),\Blast(v),\Tfirst(v),\Tlast(v)$, their $\lam(\cdot)$ and $L(\cdot)$ values.
As a point is adjacent to at most three layers, all four layer numbers can be stored on just $\log{n}+\Oh(1)$ bits.
Coordinates allow us to check for distance 1, distance two is checked by using Property~\ref{Prop:Dist2},
and larger distances by Property~\ref{Prop:Diff2}, Lemma~\ref{Lem:Lambda} and~\ref{Lem:Dist3}.

\subsection{Auxiliary points}

\begin{figure}[h]
\begin{center}
  \includegraphics[scale=1]{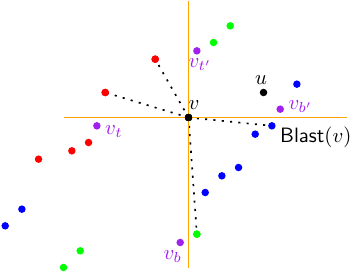}
\end{center}
\caption{Four points added for a point $v$: $v_{b}$ and $v_{b'}$ on bottom boundary, $v_{t}$ and $v_{t'}$ on top.
We have a property that whenever point $u$ is adjacent to $v_{b'}$, it was already adjacent to $\Blast(v)$.}
\label{Fig:Add4}
\end{figure}

To better manage detecting points at distance 1 without explicitly storing coordinates,
we will add, for each point in the set $S$, four additional artificial points, two on each boundary.
Consider $v$ from the initial set and its bottom-right quadrant $\BR_v$.
We add two points $v_{b}=(v_x-\epsilon, \Bfirst(v)_y-\epsilon)$ and $v_{b'}=(\Blast(v)_x+\epsilon,v_y+\epsilon)$ to the set $S$ of points.
See Figure~\ref{Fig:Add4}.
Then, we change the numeration of coordinates so that we still use the permutation of natural numbers up to $|S|$.
This is repeated for all the initial points.
First, we check that this addition did not disturb the properties of the points too much:

\begin{lemma}
All added points are on the bottom boundary.
Moreover, for any two points $u,w \in S$, $d(u,w)$ remains the same.
\end{lemma}
\begin{proof}
Considering the first property, we need to observe that when adding a point, its bottom-right quadrant is empty.
For $v_{b'}$ it holds as the point is between $\Blast(v)$ and the next point on the bottom boundary on both axes.
Thus it changes the status of no point on the bottom boundary and itself is on this boundary.
We have a similar situation with $v_{b}$.

For the second property, we notice that any point adjacent to $v_{b'}$ is also adjacent to $\Blast(v)$.
Since $v_{b'}$ and $\Blast(v)$ lie on the bottom boundary, any adjacent point must be in their top-left quadrant.
As $v_{b'}=(\Blast(v)_x+\epsilon,v_y+\epsilon)$ and there are no points with $x$-coordinate between $\Blast(v)_x$ and $\Blast(v)_x+\epsilon$,
if some point is to the left of $v_{b'}$, it is also to the left of $\Blast(v)$,
and by definition we have $v_{b'}>v_y>\Blast(v)_y$.
Similarly, any point adjacent to $v_{b}$ is also adjacent to $\Bfirst(v)$.
Therefore, these points cannot offer any shortcuts in existing shortest paths.
\end{proof}

Similarly, for each point in the initial set, we add two points on the upper boundary.
That is, consider $v$ and $\TL_v$.
We add two points $v_{t}=(\Tfirst(v)_x-\epsilon,v_y-\epsilon)$ and $v_{t'}=(v_x+\epsilon, \Tlast(v)_y+\epsilon)$ to the set $S$ of points.
Then, we again change the numeration of coordinates.
This is symmetric and has the same properties.

\begin{figure}[h]
\begin{center}
  \includegraphics[scale=1.0]{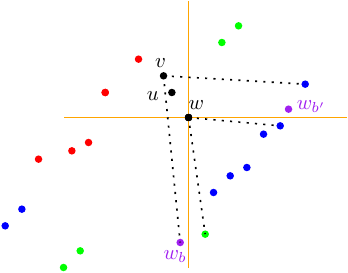}
\end{center}
\caption{We have $w \in \BR_v$, and $\Bfirst(v) < \Bfirst(w) < \Blast(w) < \Blast(v)$.
Meanwhile, $u \notin \BR_w$, as $\Blast(u)=w_{b'}>\Blast(w)$.
Auxiliary points which would be added for $u,v$ are not shown to avoid clutter.}
\label{Fig:Incl}
\end{figure}

After adding four auxiliary points for all initial points, we have the desired property:
\begin{lemma}
For any two points $v,w$ from the initial set, $w \in \BR_v$ is equivalent to $\Bfirst(v) < \Bfirst(w) \leq \Blast(w) < \Blast(v)$.
Moreover, $w \in \TL_v$ is equivalent to $\Tfirst(v) < \Tfirst(w) \leq \Tlast(w) < \Tlast(v)$.
\label{Lem:Dist1}
\end{lemma}
\begin{proof}
The cases for both boundaries are symmetrical, so we focus on the bottom one.

If $w \in \BR_v$, then by transitivity $v$ is adjacent to $\Bfirst(w)$, and then by definition of $w_b$, $v$ is also adjacent to $w_{b}$.
As $w_{b}<\Bfirst(w)$, we get $\Bfirst(v)<\Bfirst(w)$.
Analogous facts hold for $w_{b'}$, therefore implication in the right direction holds.

We consider the left direction and use contraposition.
Firstly, we want to show that if $w \notin \BR_v$, then either $\Bfirst(v)>\Bfirst(w)$ or $\Blast(w)>\Blast(v)$.
$w \notin \BR_v$ means $w_x < v_x$ or $w_y> v_y$. Assume $w_x < v_x$ and $\Bfirst(v) \leq \Bfirst(w)$.
This can be only if $\Bfirst(v) = \Bfirst(w)$, since $w$ is to the left of $v$ and points on boundaries are ordered.
As $v_{b}$ is below $\Bfirst(v)$ and between $w_x$ and $v_x$, it must be that $v_{b} \in \BR_w$.
So we have $v_{b} \in \BR_w$ and $v_{b} < \Bfirst(v)=\Bfirst(w)$, a contradiction, as then $v_{b}$ should be $\Bfirst(w)$.
Similarly using $v_{b'}$ we can show that assuming $w_y > v_y$ and $\Blast(w) \leq \Blast(v)$ leads to a contradiction.
See Figure~\ref{Fig:Incl}.
\end{proof}

The above lemma is useful when testing for adjacency of points -- we do not need explicit coordinates to check it.
Now, we are able to create labels of length $5\log{n}+\Oh(1)$,
as there is no longer a need to store coordinates, thus just four $\lam$ values,
and additionally layer numbers using only $\log{n}+\Oh(1)$ bits.
We can still improve on this, getting our final result in the next Section.

\section{Final Scheme of Size 3$\log{n}$}
\label{Sec:3log}

In this Section, the final improvement to label sizes is achieved, by collapsing two \emph{pairs} of $\lam$ values into
just two values, with an additional constant number of bits.

\begin{lemma}
We can store for each input point $v$ two integers $v_{x'},v_{y'}$ with values in $\Oh(n)$, bit values $v_{binf},v_{tinf}$,
and layer numbers $L(\Bfirst(v)),L(\Blast(v)),L(\Tfirst(v)),L(\Tlast(v))$ smaller than $n$,
such that distance queries for pairs of points can be answered using these values only.
\end{lemma}
\label{Lem:Main}

\begin{figure}[h!]
\begin{center}
  \includegraphics[scale=0.9]{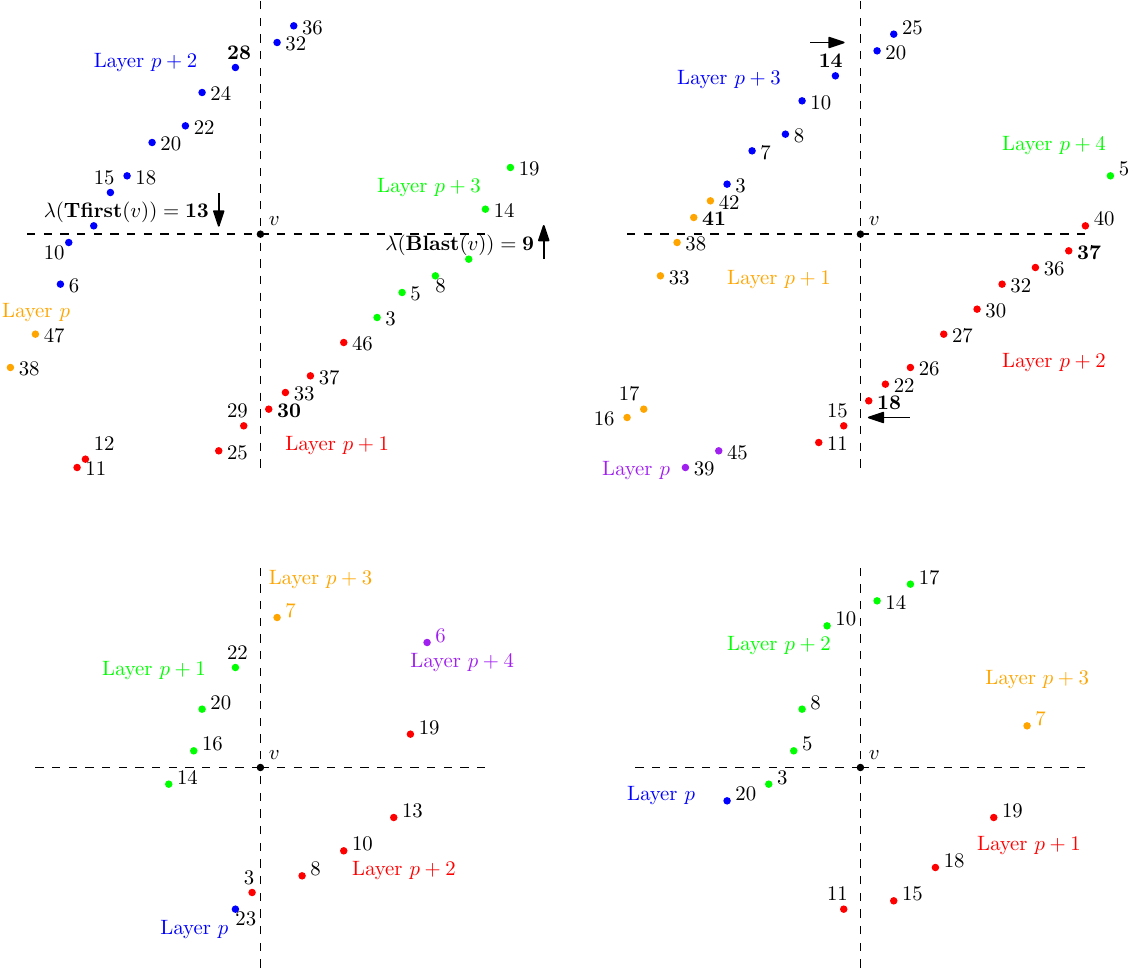}
\end{center}
\caption{In the first case, label of $v$ could store $\lam(\Bfirst(v)),\lam(\Blast(v)),\lam(\Tfirst(v)),\lam(\Tlast(v))$ 
values, which would be 30, 9, 13, and 28.
But instead, we can use in its label just two values, $v_{y'}=13-\eps$ and $v_{x'}=30-\eps$.
In the second case, point $v$ could store values 18, 37, 41, 14.
Instead, we can use $v_{y'}=40-\eps$ and $v_{x'}=18-\eps$.
\newline In the last case, we have point $v$ adjacent to only two layers,
and the lower one is on the bottom boundary.
$v$ could store values 15, 19, 5, and 10.
Instead, we can use in its label just two numbers, $v_{y'}=5$, $v_{x'}=14-\eps$, with two more bits indicating this case.
We use the fact that a point with $\lam$ value of 19 is necessarily the last one in its layer
and we won't need its exact value.}
\label{Fig:Stored}
\end{figure}

\begin{proof}
Let us consider any input point $v$, recall we made sure it does not lie on the boundary.
As previously, we can store $L(\Bfirst(v))$, $L(\Blast(v))$, $L(\Tfirst(v))$ and $L(\Tlast(v))$ on $\log{n}+\Oh(1)$ bits,
as there are no more than $n$ layers, and the differences between these four values are at most $2$.
We would like to store one number instead of $\lam(\Blast(v))$ and $\lam(\Tfirst(v))$, then also one number instead of
$\lam(\Bfirst(v))$ and $\lam(\Tlast(v))$.
We need to consider four possible cases for the layout of layers, see Figure~\ref{Fig:Stored}:
\begin{enumerate}
\item $v$ is adjacent to two layers on the bottom boundary (and then necessarily one on the top).
\item $v$ is adjacent to two layers on the top boundary.
\item $v$ is adjacent to one layer on both boundaries, and the layer on the bottom is higher.
\item $v$ is adjacent to one layer on both boundaries, and the layer on the bottom is lower.
\end{enumerate}
These will be referred to as possible layouts.

First, consider the last case.
We argue that it must be that $\Blast(v)$ is the last point in its layer.
Assume otherwise, so there is a point $w$ on the bottom boundary, with $L(w)=L(\Blast(v))$ and $w>\Blast(v)$,
so necessarily $w \in TR_v$.
This means that there is a point $u$ with $L(u)=L(w)-1$ and $u_y>w_y>v_y$, by definition of layers.
It cannot be $u \in TL_v$, as we assumed the last case, where there are only points from layer $L(u)+2$ in $TL_v$.
But if $u \in TR_v$, and thus $u_x>v_x$, then no point from layer $L(w)+1$ could be in $TL_v$,
as all of them are to the right of $u$, and we reach a final contradiction.
Therefore, $\Blast(v)$ is in this case the last point in its layer.
Notice that this information is easy to store, as it was proven that the last point in a layer has $\lam$ value larger than
all of the points in larger layers, that is, there is always a quick path to such points.
We can store a bit $v_{binf}=1$ indicating this case, effectively using value of infinity instead of exact $\lam(\Blast(v))$,
and then just store unchanged $\lam(\Tfirst(v))$.
We still need to argue that no point from lower layers must use $\Blast(v)$ to reach $v$, and that
checks for distances 1 and 2 works, which will be done later.

Now, consider the three first cases.
In all of them, we have $L(\Blast(v))>L(\Tfirst(v))$, and as these two points are adjacent
$\lam(\Blast(v))<\lam(\Tfirst(v))$ also holds.
Denote by $w$ a boundary point with the smallest $\lam(w)$ value among points with $w_y>v_y$
and $\lam(w)>\lam(\Blast(v))$.
Note that as $\Tfirst(v)$ can be chosen, such $w$ always exists and it must be $\lam(w) \leq \lam(\Tfirst(v))$.
We will store in our label value of $v_{y'}=\lam(w)-\eps$ (to be normalised later),
and as we will see this one number can replace both values of $\lam(\Blast(v)),\lam(\Tfirst(v))$.
We say that value $\lam(\Blast(v))$ is increased to $v_{y'}$, and $\lam(\Tfirst(v))$ is decreased.

\begin{algorithm}[hbt!]
\begin{algorithmic}[1]
  \Function{Encode}{$S$}
  \State \textbf{Input:} set of points $S$, representing permutation graph
  \State \textbf{Output:} labels for all points in $S$
  \Statex
  \State $S' \gets S$
  \For{\texttt{$v \in S$}} \Comment Loop for adding auxiliary points
    \State Add to $S'$ $v_{b}, v_{b'}, v_{t}, v_{t'}$
    \If{$v$ is on the boundary of $S'$}
      \State Add to $S'$ point removing $v$ from boundary
    \EndIf
  \EndFor  
  \State Compute layer numbers for boundary points of $S'$
  \State Compute $\lam$ numbers for boundary points of $S'$
  \For{\texttt{$v \in S$}}
    \If{$\Blast(v)$ is last in its layer} \Comment Check for the special case
      \State $v_{binf} \gets 1$
      \State $v_{y'} \gets \lam(\Tfirst(v))$
    \Else
      \State $v_{binf} \gets 0$
      \State Find boundary $w$ with $w_y>v_y$, $\lam(w)>\lam(\Blast(v))$, and minimum $\lam(w)$
      \State $v_{y'} \gets \lam(w)-\eps$ \Comment Collapsed value for $\lam(\Blast(v)),\lam(\Tfirst(v))$
    \EndIf
    \If{$\Tlast(v)$ is last in its layer}
      \State $v_{tinf} \gets 1$
      \State $v_{x'} \gets \lam(\Bfirst(v))$
    \Else
      \State $v_{tinf} \gets 0$
      \State Find boundary $w'$ with $w'_x>v_x$, $\lam(w')>\lam(\Tlast(v))$, and minimum $\lam(w')$
      \State $v_{x'} \gets \lam(w')-\eps$ \Comment Collapsed value for $\lam(\Tlast(v)),\lam(\Bfirst(v))$
    \EndIf
    \State Store in $\ell(v)$ $L(\Bfirst(v))$ 
    \State Store in $\ell(v)$ constant number of bits encoding $L(\Blast(v)), L(\Tfirst(v)), L(\Tlast(v))$
    \State Store in $\ell(v)$ $v_{binf}, v_{y'}, v_{tinf}, v_{x'}$
    \State Output $\ell(v)$
  \EndFor
  \EndFunction
\end{algorithmic}
\caption{Encoder computing labels for a set of points representing permutation graph.}
\label{Alg:Encode}
\end{algorithm}

We can deal in a similar manner with $\Bfirst(v),\Tlast(v)$ values.
If $\Tlast(v)$ is the last point in its layer (which is always true in the third case of the possible layouts),
we store $v_{tinf}=1$ and an exact value of $\lam(\Bfirst(v))$.
Otherwise, denote by $w'$ a boundary point with the largest $\lam(w')$ among points with $w'_x>v_x$
and $\lam(w')>\lam(\Tlast(v))$.
We will store value of $v_{x'}=\lam(w')-\eps$, and this one value can replace both $\lam(\Tlast(v)),\lam(\Bfirst(v))$.
See Algorithm~\ref{Alg:Encode} for the summary of the encoder work.

We denote by $d((l,i),u)$, for numbers $l,i$ and boundary point $u$, a distance between $u$ and point in layer $l$
with value of $\lam$ equal to $i$, as would be returned by using Lemma~\ref{Lem:Lambda}.
Decoding navigates all possible cases and makes some distance queries, using values of $v_{x'},v_{y'}$.
See Algorithm~\ref{Alg:Decode} for the description of the decoder, some details will become clear later.

\begin{algorithm}[hbt!]
\begin{algorithmic}[1]
  \Function{D}{$\ell(v),\ell(u)$}
  \State \textbf{Input:} labels of two points from the same permutation graph.
  \State \textbf{Output:} distance between points in the graph.
  \Statex
  \Function{Extract}{$\ell(t)$} \Comment Obtaining necessary values from label of $t$
  \State Extract $L(\Bfirst(t)), L(\Blast(t)), L(\Tfirst(t)), L(\Tlast(t))$ from $\ell(t)$
  \State $\lam'(\Tfirst(t)) \gets t_{y'}$ \Comment $\lam'$ are modified values, but to be treated as real $\lam$
  \If{$t_{binf} = 1$}
    \State $\lam'(\Blast(t)) \gets \infty$
  \Else
    \State $\lam'(\Blast(t)) \gets t_{y'}$
  \EndIf
  \State $\lam'(\Bfirst(t)) \gets t_{x'}$
  \If{$t_{tinf} = 1$}
    \State $\lam'(\Tlast(t)) \gets \infty$
  \Else
    \State $\lam'(\Tlast(t)) \gets t_{x'}$
  \EndIf
  \EndFunction
  \Statex
  \State \Call{Extract}{$\ell(u)$}, \Call{Extract}{$\ell(v)$}
  \If{range of $u$ or $v$ is a subset of the range of the other, on any boundary,}
    \State \Return 1
  \EndIf
  \If{ranges of $u$ and $v$ intersect on any boundary}
    \State \Return 2
  \EndIf
  \State $tmp \gets \infty$ \Comment We are left with the case of distance $>2$
  \For {\texttt{$i \in \{\Blast(v),\Tlast(v)\}$}}
    \For {\texttt{$j \in \{\Bfirst(u),\Tfirst(u)\}$}}
      \State $tmp \gets \min(tmp,d([L(i),\lam'(i)] , [L(j),\lam'(j)]))$
    \EndFor
  \EndFor
  \For {\texttt{$i \in \{\Blast(u),\Tlast(u)\}$}}
    \For {\texttt{$j \in \{\Bfirst(v),\Tfirst(v)\}$}}
      \State $tmp \gets \min(tmp,d[(L(i),\lam'(i)] , [L(j),\lam'(j)]))$
    \EndFor
  \EndFor
  \State \Return $d$
  \EndFunction
\end{algorithmic}
\caption{Decoder checking the distance between two points given only their labels.}
\label{Alg:Decode}
\end{algorithm}

To prove correctness, first consider distances of at least three.

\begin{property}
For $u,v$ with $d(u,v) \geq 3$, our scheme will return $d(u,v)$, or a value less than 3.
\label{Prop:Dec3}
\end{property}
\begin{proof}
The formulation is for technical reasons, we will exclude the possibility of returning value less than 3 later.
For some vertex $v$, due to distance, we are interested only in paths from $v$ to points in $\BL_v$ and $\TR_v$.
By Lemma~\ref{Lem:Dist3}, the first ones start in either $\Bfirst(v)$ or $\Tfirst(v)$,
other ones in either $\Blast(v)$ or $\Tlast(v)$.
Let us informally go through what we need to check for any vertex $v$:
\begin{itemize}
\item For any boundary point $p \in \TR_v$ with $d(v,p) \geq 3$, distances between $p$ and $\Blast(v),\Tlast(v)$ are preserved
by the encoding.
\item For any boundary point $p \in \BL_v$ with $d(v,p) \geq 3$, distances between $p$ and $\Blast(v),\Tlast(v)$
as checked by the decoder are not smaller than in the graph (but are allowed to be larger).
\item For any boundary point $p \in \BL_v$ with $d(v,p) \geq 3$, distances between $p$ and $\Bfirst(v),\Tfirst(v)$ are preserved.
\item For any boundary point $p \in \TR_v$ with $d(v,p) \geq 3$, distances between $p$ and $\Bfirst(v),\Tfirst(v)$
as checked by the decoder are not smaller than in the graph.
\end{itemize}

First, consider $\Blast(v)$.
We need that for any boundary point $u$ in $\TR_v$, $d(\Blast(v),u)=d((L(\Blast(v)),v_{y'}), u)$.
This holds by choice of $v_{y'}$, no relation between $\lam$ values has changed.
In the case of the last layout, using infinite value also preserves these relations.

We also need to consider whether we could have reduced distance to some points in $BL_v$
by increasing stored $\lam(\Blast(v))$ value to $v_{y'}$.
As for a quick path to exists relations between layer numbers and $\lam$ values need to be opposite,
increasing $\lam(\Blast(v))$ could introduce a false quick path only for points in layers larger than $L(\Blast(v))$.
For $BL_v$ this is possible only in the last layout, when $\Blast(v)$ is the last point in its layer
and is adjacent to all points in layer $L(\Blast(v))+1$ anyway, by Lemma~\ref{Lem:Lambda}.

\begin{figure}[h]
\begin{center}
  \includegraphics[scale=0.8]{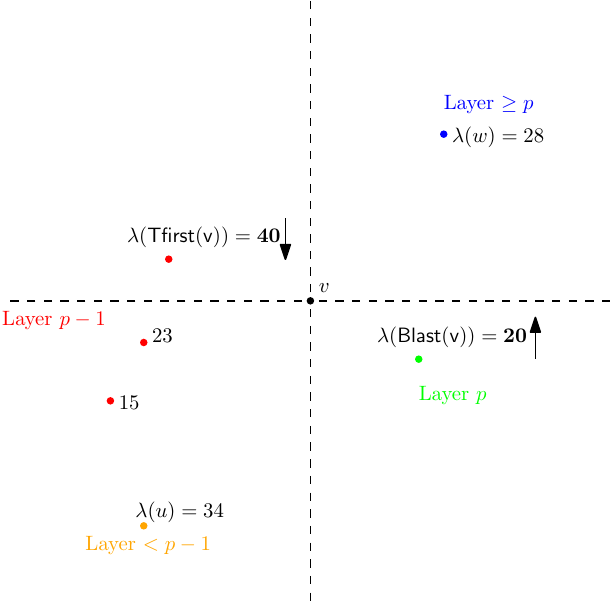}
\end{center}
\caption{$w$ is defined as the point above $v$ with the smallest $\lam$ larger than $\lam(\Blast(v))$.
We show it is impossible for a point $u \in \BL_v$ to have $\lam(u)>\lam(w)$, as it cannot have a single point in layer
$L(\Tfirst(v))$ on a quick path to $w$.}
\label{Fig:Dist3}
\end{figure}

We had a $\Blast(v)$, now consider $\Tfirst(v)$, where the situation is more complicated due to the non-symmetric definition of $v_{y'}$.
We assume $v_{binf}=0$, as otherwise the exact value of $\lam(\Tfirst(v))$ is stored and the decoder does not err.
We need that for any boundary point $u$ in $\BL_v$, $d(\Tfirst(v),u)=d((L(\Tfirst(v)),v_{y'}), u)$.
By examining possible layouts we see that any point $u \in BL_{v}$ with $L(u)>L(\Tfirst(v))$ must have
$L(u)=L(\Blast(v))$, thus $\lam(u)<\lam(\Blast(v))$, so $\lam(u)<v_{y'}<\lam(\Tfirst(v))$ and the distance query is still correct.
If $L(u)=L(\Tfirst(v))$, distance is always 2 no matter $\lam$ values.

Now assume there is boundary point $u \in \BL_{v}$ with $L(u)<L(\Tfirst(v))$ and $v_{y'}<\lam(u)<\lam(\Tfirst(v))$,
which is the only remaining way to produce false result of distance query.
Note that $L(\Tfirst(v))<L(\Blast(v))$, by $v_{binf}=0$.
Recall $w$ is a boundary point with the smallest $\lam(w)$ value larger than $\lam(\Blast(v))$ among points with $w_y>v_y$
thus we have $\lam(\Tfirst(v)) \geq \lam(w) > \lam(\Blast(v))$.
If $L(w)=L(\Tfirst(v))$, then $w=\Tfirst(v)$ and we have nothing to prove, so we assume $L(w)>L(\Tfirst(v))$.
See Figure~\ref{Fig:Dist3}.
It holds that $\lam(u)>\lam(w)$ and $L(u) < L(\Tfirst(v)) < L(w)$, thus $d(u,w)=L(w)-L(u)$.
This means that there is a quick path $P$ from $u$ to $w$ containing exactly one point from each layer in range $[L(u),L(w)]$.

Let $r \in P$ and $L(r)=L(\Tfirst(v))$.
There is no quick path from $u$ to $\Tfirst(v)$, so it must be $r<\Tfirst(v)$, which means $r_y < v_y$.
This in turn means that $r$ is not adjacent to any point in layer $L(\Blast(v))$
with $\lam$ value larger than $\lam(\Blast(v)$.
But then it cannot be that $P$ contains a single point from layer $L(\Blast(v))$, as such a point must be adjacent to
$r$ and also have $\lam$ value of at least $\lam(w)$, meaning larger than $\lam(\Blast(v)$.
Therefore, by contradiction, there cannot be a point $u \in BL_{v}$ with $L(u)<L(\Tfirst(v))$
and $v_{y'}<\lam(u)<\lam(\Tfirst(v))$, which means that for all $u \in \BL_v$
we have $d(\Tfirst(v),u)=d((L(\Tfirst(v)),v_{y'}), u)$.

We also need to consider whether we could have reduced distance to some points in $\TR_v$
by using $v_{y'}$ value instead of $\lam(\Tfirst(v))$.
But in $\TR_v$ there are only points in layers at least as large as $L(\Tfirst(v))$,
for which decreasing value of $\lam(\Tfirst(v))$ cannot decrease output of distance query.

Proof for $v_{x'}$, $\Tlast(v)$, and $\Bfirst(v)$ is symmetric.
\end{proof}

As a side note, let us observe two things that may help in understanding our methods.
It might be $d((L(\Blast(v)),v_{y'}), u) = d(\Blast(v), u)+2$ for point $u \in BL_{v}$,
but by Lemma~\ref{Lem:Dist3} we do not need to store this distance correctly, as there is a shortest path from $u$
to $v$ with the penultimate point not being $\Blast(v)$.
In other words, we do lose some unnecessary information.
Secondly, we can store infinity values for boundary points that are last in their layers because of
layers and $\lam$ values definition -- all points in the following layers have smaller $\lam$ values already.
We could not store 'zero' or 'minus infinity' values for points that are first in their layers,
but there is no need for this.

Now we turn to distances 1 and 2, with the former being easier.

\begin{property}
For $u,v$ with $d(u,v) = 1$, our scheme will return $1$.
For $u,v$ with $d(u,v) \geq 2$, our scheme will never return 1.
\label{Prop:Dec1}
\end{property}
\begin{proof}

Once again, we have several cases of $\lam$ values being increased, decreased, or set to infinity, but can check that
methods from Lemma~\ref{Lem:Dist1} and Property~\ref{Prop:Dist2} still holds for our new values.
The encoder never changes layers of points, so any possible issue is connected only to $\lam$ values.
Consider input points $u,v$ with $v_y > u_y$ and $L(\Blast(v))=L(\Blast(u))$.
By $v_y > u_y$, we have $\Blast(v) \geq \Blast(u)$.
Then if $v_{binf}=u_{binf}=0$, we must have $v_{y'} \geq u_{y'}$.
Indeed, when $\Blast(v) \neq \Blast(u)$, then $u_{y'} < \lam(\Blast(v))<v_{y'}$.
If $\Blast(v) = \Blast(u)$, again $v_{y'} \geq u_{y'}$ as the smallest value of $\lam$ larger than $\lam(\Blast(v))$
for points above $u$ cannot be larger than analogous value for $v$ by $v_y > u_y$.
Finally, $u_{binf}=1$ implies $v_{binf}=1$.
This means that for any $u,v$, we can derive whether $\lam(\Blast(v)) \geq \lam(\Blast(u))$ given $\ell(v),\ell(u)$
simply by comparing values from the labels.
As $\Tfirst(v)$ too is replaced by $v_{y'}$, relation between $\Tfirst(v),\Tfirst(u)$ is retained.
Similarly, we can check that respective inequalities hold for $v_{x'}, u_{x'}$ values.

Overall, this means that relations between points of the same kind (say $\Bfirst(v),\Bfirst(u)$) are retained,
therefore by Lemma~\ref{Lem:Dist1} our labeling correctly outputs 1 exactly when $d(u,v) = 1$.
\end{proof}

Lastly, we are left with the case of distance two.

\begin{property}
For $u,v$ with $d(u,v)=2$, our scheme will return 2.
For $u,v$ with $d(u,v) \geq 3$, our scheme will never return 2.
\label{Prop:Dec2}
\end{property}
\begin{proof}

Using Property~\ref{Prop:Dist2}, whenever $d(u,v)=2$ the decoder reports this correctly, as values of $\lam(\Blast),
\lam(\Tlast)$ can only be increased, and $\lam(\Bfirst),\lam(\Tfirst)$ only decreased, and thus ranges only widen.
Therefore, we need to just exclude the possibility of false intersections, that is, reporting 2 for $d(u,v) > 2$.
First let us note that the case when $\lam(\Blast(u))<\lam(\Bfirst(v))$ but $u_{y'}=v_{x'}$ is impossible to achieve.
This is because whenever $\lam$ value is increased or decreased, it is set to some new unique value, different from all
values existing at the moment, and $\lam(\Blast(u)),\lam(\Bfirst(v))$ are not changed simultaneously.

\begin{figure}[h]
\begin{center}
  \includegraphics[scale=1]{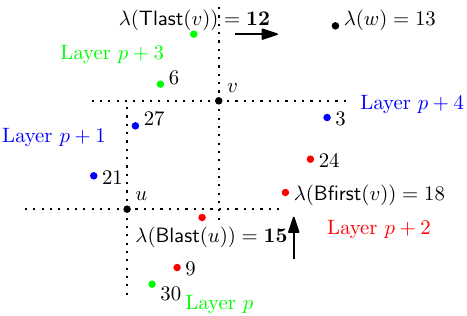}
\end{center}
\caption{Here we consider, for $v>u$ with non-intersecting neighbourhoods on the bottom boundary
and $L(\Blast(u))=L(\Bfirst(v))$, whether we could get $u_{y'}>v_{x'}$,
which would mean a false intersection in the decoding process.
We show that it is impossible for some boundary point $w$ to the right of $v$ to have $\lam(\Tlast(v))<\lam(w)<\lam(\Blast(u))$.}
\label{Fig:Nope}
\end{figure}

Consider input points $v,u$ with $v_y > u_y$, $L(\Blast(v))=L(\Bfirst(u))$, $v_{tinf}=u_{binf}=0$, and $d(u,v)>2$.
From distance constraint, we get that also $v_x > u_x$.
Assume $\Blast(u)<\Bfirst(v)$, meaning ranges of neighbours of $v,u$ on the bottom layer do not intersect. 
Recall that the encoder replaces $\lam(\Blast(u))$ with $u_{y'}$, and $\lam(\Bfirst(v))$ with $v_{x'}$,
where $v_{x'}$ is defined using $\lam(\Tlast(v))$.
We will show that it is impossible that $u_{y'}>v_{x'}$, and thus ranges of $v,u$ still do not intersect.

First, assume $\lam(\Tlast(v))>\lam(\Blast(u))$.
Then, by definition $u_{y'}<\lam(\Tlast(v))$ as $\Tlast(v)$ is above $u$, and $v_{x'}>\lam(\Tlast(v))$,
which means $u_{y'}<v_{x'}$ as needed.
So, we might assume $\lam(\Tlast(v))<\lam(\Blast(u))<\lam(\Bfirst(v))$.
Now, for $u_{y'}>v_{x'}$ to hold, we would need a boundary point $w$ somewhere to the right of $v$,
and with $\lam(\Tlast(v))<\lam(w)<\lam(\Blast(u))$.
By definitions there are no points simultaneously below $u$ and to the right of $\Blast(u)$,
and by $\Blast(u)<\Bfirst(v)$ and $v_y > u_y$, $v$ is to the right of $\Blast(u)$.
Thus, $w$ must be above $u$.
We assumed $v_{tinf}=0$, excluding the third case of possible layouts (Figure~\ref{Fig:Stored}) $v$,
so $L(\Tlast(v))-1=L(\Bfirst(v))=L(\Blast(u))$.
By $\lam(w)<\lam(\Blast(u))$, it also must be that $L(w) \geq L(\Tlast(v))$.
Summing up, we get that $L(\Blast(u))=L(\Bfirst(v))<L(\Tlast(v)) \leq L(w)$.
See Figure~\ref{Fig:Nope}, which depicts the only relevant remaining arrangement of points.

We examine $\Blast(u)$ to show that this case is also impossible.
Since $L(w)>L(\Blast(u))$ and $\lam(w)<\lam(\Blast(u))$, it holds that $d(\Blast(u),w)=L(w)-L(\Blast(u))$
and there is a quick path between these two points, having a single point in each of layers $[L(\Blast(u)),L(w)]$.
But as $\lam(\Tlast(v))<\lam(w)$ and $L(\Tlast(v)) \leq L(w)$, we get $d(\Tlast(v),w)=L(w)-L(\Tlast(v))+2$.
As $\Blast(u)$ is to the left of $v$, the largest adjacent point in layer $L(\Tlast(v))$ has $\lam$ value at most
$\lam(\Tlast(v))$, so smaller than $\lam(w)$.
But this is a contradiction, as then there cannot be a quick path from $\Blast(u)$ to $w$
having a single point in layer $L(\Tlast(v))$.

This means that we can still use Property~\ref{Prop:Dist2} for new stored values.
\end{proof}

By Properties~\ref{Prop:Dec3}, \ref{Prop:Dec1} and \ref{Prop:Dec2}, we have proven the lemma.
\end{proof}

To conclude, we have that $\ell(v)$ consists of the following parts:
\begin{enumerate}
\item $L(\Bfirst(v))$, $L(\Blast(v))$, $L(\Tfirst(v))$ and $L(\Tlast(v))$, all stored on total $\log{n}+\Oh(1)$ bits
due to differences between these values being at most 2.
\item Bit $v_{binf}$ and value $v_{y'}$, on $\log{n}+\Oh(1)$ bits.
\item Bit $v_{tinf}$ and value $v_{x'}$, like above.
\end{enumerate}

We increased the number of vertices in the graph by a constant factor, so the final length is $3\log{n}+\Oh(1)$ bits.
Decoding can be done in constant time.

\subsection{Disconnected graphs}
Here we describe how to modify distance labeling for connected graphs into distance labeling for general graphs,
by adding at most $\Oh(\log{\log{n}})$ bits to the labels.
This is a standard simple approach, present in some related works.
We sort connected components of the graph by decreasing size, say we have $C_1,C_2,\ldots$,
then proceed with creating distance labeling for each individual connected component.
The final label is the number of connected component of a vertex and then its label for the distance created in the component.
If $v \in C_i$, we encode $i$ on $\log{i}+\Oh(\log{\log{n}})$ bits.
We have $|C_i| \leq n/i$, so the final label size is at most $\log{(n/i)}+\Oh(\log{\log{n}})+\log{i}=\log{n}+\Oh(\log{\log{n}})$, as claimed.

\section{Conclusion}
Improving upon the previous results, we have described a distance labeling scheme for permutation graphs matching existing
lower bound up to an additive second-order $\Oh(\log{\log{n}})$ term.
This also improves constants in distance labeling for circular permutation graphs, as described in~\cite{BazzaroG05}.
Namely, one can construct distance labeling of size $6\log{n}+\Oh(\log{\log{n}})$ for such graphs.
We leave as an open question determining the complexity of distance labeling for circular permutation graphs,
and finding more interesting generalisations.


\bibliographystyle{plain}
\bibliography{biblio}

\end{document}